\documentclass[a4paper,USenglish,10pt]{article}
\usepackage[utf8]{inputenc}

\usepackage{amsthm,amssymb,amsfonts,amsmath}
\usepackage{amssymb,amsfonts,amsmath}
\newtheorem{theorem}{Theorem}
\newtheorem{lemma}[theorem]{Lemma}
\newtheorem{corollary}[theorem]{Corollary}
\newtheorem{proposition}[theorem]{Proposition}

\theoremstyle{definition}

\newtheorem{problem}{Open Problem}

\usepackage[hidelinks]{hyperref}
\usepackage[blocks]{authblk}

\usepackage{graphicx}
\graphicspath{{figures/}}
\usepackage{wrapfig}

\newcommand{\fig}[1]{\figurename~\ref{#1}}

\makeatletter
\g@addto@macro\bfseries{\boldmath}
\makeatother

\usepackage{mathtools}
\DeclarePairedDelimiter{\ceil}{\lceil}{\rceil}

\newcommand\blfootnote[1]{%
  \begingroup
  \renewcommand\thefootnote{}\footnote{#1}%
  \addtocounter{footnote}{-1}%
  \endgroup
}

\usepackage[mathlines]{lineno}
 \newcommand*\patchAmsMathEnvironmentForLineno[1]{%
   \expandafter\let\csname old#1\expandafter\endcsname\csname #1\endcsname
   \expandafter\let\csname oldend#1\expandafter\endcsname\csname end#1\endcsname
   \renewenvironment{#1}%
      {\linenomath\csname old#1\endcsname}%
      {\csname oldend#1\endcsname\endlinenomath}}%
 \newcommand*\patchBothAmsMathEnvironmentsForLineno[1]{%
   \patchAmsMathEnvironmentForLineno{#1}%
   \patchAmsMathEnvironmentForLineno{#1*}}%
 \AtBeginDocument{%
 \patchBothAmsMathEnvironmentsForLineno{equation}%
 \patchBothAmsMathEnvironmentsForLineno{align}%
 \patchBothAmsMathEnvironmentsForLineno{flalign}%
 \patchBothAmsMathEnvironmentsForLineno{alignat}%
 \patchBothAmsMathEnvironmentsForLineno{gather}%
 \patchBothAmsMathEnvironmentsForLineno{multline}%
 }

\advance\hoffset-7mm
\date{}
\begin{document}

\title{On Plane Subgraphs of Complete Topological Drawings}

\author[1]{Alfredo Garc\'ia\thanks{Supported by MINECO project MTM2015-63791-R and Gobierno de Arag\'on under Grant E41-17 (FEDER).}}
\author[2]{Alexander Pilz\thanks{Supported by a Schr\"odinger fellowship of the Austrian Science Fund (FWF): J-3847-N35.}}
\author[1]{Javier Tejel$^*$
}
\affil[1]{Departamento de M\'etodos Estad\'isticos and IUMA, Universidad de Zaragoza.
  \texttt{olaverri@unizar.es, jtejel@unizar.es}}
\affil[2]{Institute of Software Technology, Graz University of Technology.
  \texttt{apilz@ist.tugraz.at}}

\maketitle


\begin{abstract}

Topological drawings are representations of graphs in the plane, where
vertices are represented by points, and edges by simple curves connecting the points. A drawing is \emph{simple} if two edges intersect at most in a single point, either at a
common endpoint or at a proper crossing. In this paper we study properties of maximal plane subgraphs of simple drawings $D_n$ of the complete graph $K_n$ on $n$ vertices.
Our main structural result is that maximal plane subgraphs are 2-connected and what we call \emph{essentially 3-edge-connected}. Besides, any maximal plane subgraph contains at least $\lceil 3n/2 \rceil$ edges.
We also address the problem of obtaining a plane subgraph of $D_n$ with the maximum number of edges, proving that this problem is NP-complete.
However, given a plane spanning connected subgraph of $D_n$, a maximum plane augmentation of this subgraph can be found in $O(n^3)$ time.
As a side result, we also show that the problem of finding a largest compatible plane straight-line graph of two labeled point sets is NP-complete.

\blfootnote{\begin{minipage}[l]{0.2\textwidth} \vspace{-8pt}\includegraphics[trim=10cm 6cm 10cm 5cm,clip,scale=0.15]{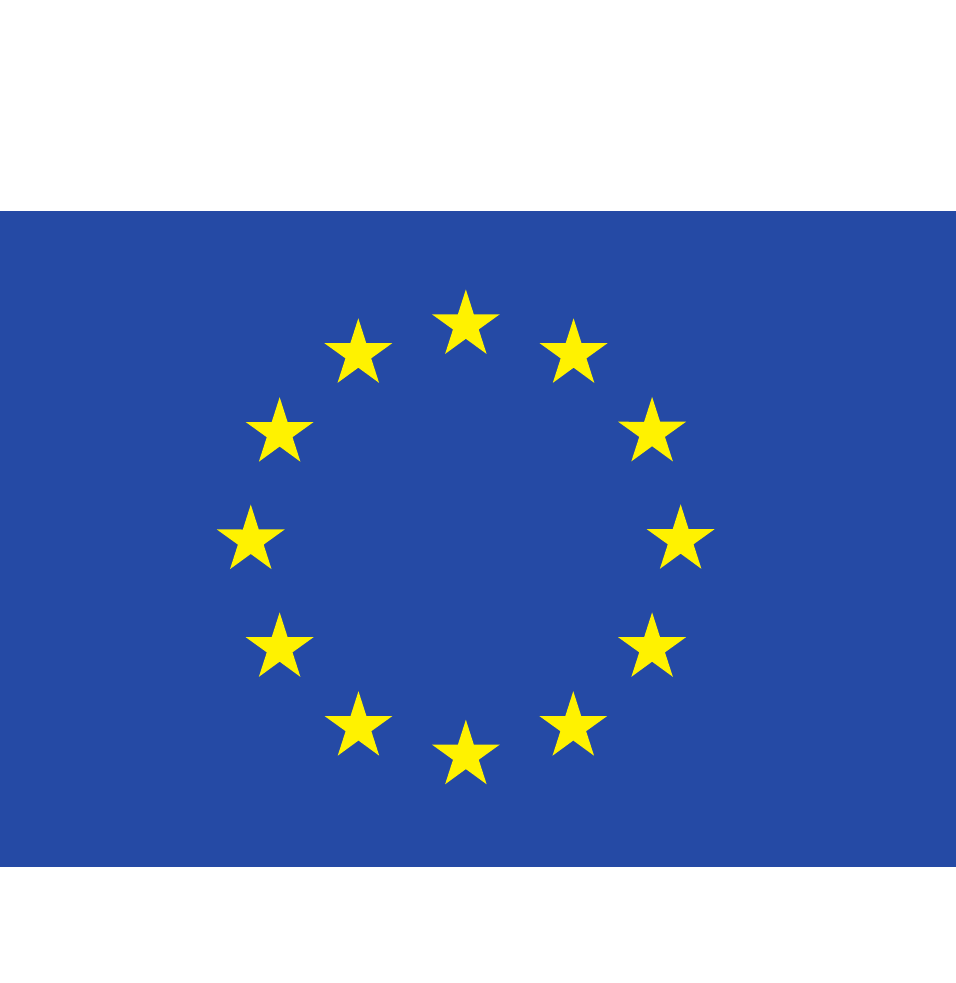} \end{minipage}  \hspace{-1.3cm} \begin{minipage}[l][1cm]{0.82\textwidth}
      This project has received funding from the European Union's Horizon 2020 research and innovation programme under the Marie Sk\l{}odowska-Curie grant agreement No 734922.
     \end{minipage}}


\textbf{\textit{Math. Subj. Class.}}: 05C10, 68R10.

\end{abstract}

\section{Introduction}
In a \emph{topological drawing} (in the plane or on the sphere) of a graph, vertices are represented by points and edges by simple curves connecting the corresponding pairs of points.
Usually, we only consider drawings satisfying some natural non-degeneracy conditions, in particular a drawing is called \emph{simple} (or a \emph{good drawing}) if two edges intersect at most in a single point, either at a common endpoint or at a crossing in their relative interior. When all the edges of a topological drawing are straight-line segments, then the drawing is
called a \emph{rectilinear drawing} or \emph{geometric graph}.

In this paper we consider only simple topological drawings of the complete graph $K_n$ on $n$ vertices. Simple topological drawings of complete graphs have been studied extensively, mainly in the context of crossing number problems. It is well known that a drawing minimizing the number of crossings has to be simple, and besides, if $n\ge 8$, the drawings of $K_n$  minimizing that crossing number are not rectilinear. We refer the reader to~\cite{ bishellable15, shellable_drawings, Balko_Monotone} for recent advances on the Harary-Hill conjecture on the minimum number of crossings of drawings of $K_n$, and to the survey~\cite{Schaefer2013} for some variants on this crossing number problem.

The problem of enumerating all the non-isomorphic drawings of $K_n$ has been studied in~\cite{ all_good_drawings, Kyncl2009, Kyncl2013, rafla} (two drawings are isomorphic if there is a homeomorphism
of the sphere that transforms one drawing
into the other).

Let $D_n$ be a simple topological drawing of $K_n$. Herein, we consider graphs in connection with their drawings, and in particular when addressing subgraphs of $K_n$ we also consider the associated sub-drawing of $D_n$. We are interested in crossing-free edge sets $F$ in $D_n$, and we will say that $F$ is a plane subgraph of $D_n$. Crossing-free edge sets in~$D_n$ have attracted considerable attention, in part because problems on embedding graphs on a set of points usually generalize to finding plane subgraphs of $D_n$. For instance, the problem of computing the maximum number of plane Hamiltonian
cycles that a simple drawings $D_n$ can contain,  is a generalization of the same problem considering only rectilinear drawings of $K_n$. And this last is the (open) problem of computing the maximum number of simple $n$-gons that can be formed on $n$ points in the plane.

There are relatively few
results on plane subgraphs of $D_n$. It is well known that in any drawing $D_n$ of $K_n$, there are plane subgraphs with $2n-3$ edges, and that there are at most $2n-2$ edges uncrossed by any other edge~\cite{fulek_ruiz_vargas,Ringel,Harborth74}. Pach, Solymosi, and T\'oth~\cite{pach_solymosi_toth} showed that any $D_n$ has $\Omega\left(\log^{1/6}(n)\right)$ pairwise disjoint edges.
This bound was subsequently improved in~\cite{fox_sudakov,pach_toth,suk}.
The current best bound of $\Omega(n^{1/2-\epsilon})$ is by Ruiz-Vargas~\cite{many_disjoint}. However, the much stronger conjecture that any simple drawing $D_n$ of $K_n$ contains a plane Hamiltonian cycle remains unproved, although it has been verified for $n\le 9$, see~\cite{all_good_drawings}.

In the course of their work on disjoint edges and empty triangles in $D_n$, Fulek and Ruiz-Vargas~\cite{fulek_ruiz_vargas} showed the following lemma.\footnote{Their lemma is actually more general.
It does not require $F$ and $v$ to be elements of a drawing of $K_n$, but rather of a drawing that contains all edges from $v$ to vertices of $F$.}

\begin{lemma}[Fulek and Ruiz-Vargas~\cite{fulek_ruiz_vargas}]\label{lem:fulek_ruiz_vargas}
Between any plane connected subgraph $F$ of $D_n$ and a vertex $v$ not in $F$, there exist at least two edges from $v$ to $F$ that do not cross~$F$.
\end{lemma}

This result can be used to build large plane subgraphs. For instance, we can begin with $F$ consisting of only one edge, then for each vertex $v$ not in $F$, we add to $F$ the edges from $v$ to $F$ not crossing $F$. In this way, we will obtain a maximal plane subgraph: a plane subgraph $\overline{F}$ such that any edge $e\notin \overline{F}$ crosses some edge of $\overline{F}$.

In Section 2 of this work, we extend that Lemma~\ref{lem:fulek_ruiz_vargas} to arbitrary (not necessarily connected) plane subgraphs. Further, in Section 3, we prove that any plane subgraph of $D_n$ can be augmented to a 2-connected plane subgraph of~$D_n$. A consequence of this result is that maximal plane subgraphs contain at least $\min(\ceil{3n/2}, 2n-3)$ edges, and this bound is tight.
Maximal plane subgraphs of $D_n$ have other interesting properties.
For example, we show that, when removing two edges from a maximal plane subgraph, it either stays connected or one of the two components is a single vertex.
Another consequence of the previous results is that for every vertex $v$ of a drawing $D_n$, there is a plane subgraph of $D_n$ consisting of the $n$-vertex star of edges incident to $v$, plus the edges of a spanning tree on the $n-1$ vertices of $V \setminus \{v\}$.

The problem setting changes when we want our plane graphs not only to be maximal, but also to contain the maximum number of edges.
While for geometric graphs, every maximal plane subgraph is a triangulation and thus also has a maximum number of edges, the situation is different for plane subgraphs of $D_n$. In Section 4, we will prove that computing a plane subgraph of $D_n$ with maximum number of edges is an NP-complete problem.
However, if a connected plane spanning subgraph $F$ is given,
we can adapt a classic algorithm from computational geometry to show that a maximum plane augmentation of $F$ can be found in $O(n^3)$ time.

As a side result, we also show that the problem of finding a largest compatible plane graph on two labeled point sets is NP-complete.

Finally, going back to Lemma~\ref{lem:fulek_ruiz_vargas}, we give an $O(n)$ algorithm to compute all the edges from a vertex $v$ to a plane connected subgraph $F$ that do not cross $F$. 

\section{Adding a single vertex}
We now discuss a generalization of Lemma~\ref{lem:fulek_ruiz_vargas} to arbitrary plane subgraphs.
This generalization will also follow independently from Theorem~\ref{thm:2-connected}.
Still, the following proposition gives further insight on the position of the uncrossed edges around the vertex~$v$, which might help in the construction of  algorithms.

We assume that a simple topological drawing~$D_n$ of $K_n$ in the plane is given, with vertex set $V = \{v_1, \dots, v_n\}$. If $x_1,x_2$ are two points on an edge $e$ of $D_{n}$ (not necessarily the endpoints of $e$), by line $x_1x_2$ we mean the portion of the curve $e$ of the drawing placed between the points $x_1$ and $x_2$.
For a vertex $v$, the \emph{star graph} with center~$v$ is the subgraph formed by the edges connecting $v$ to all the other vertices.
We denote this set of edges by $S(v)$, usually call \emph{rays} to these edges emanating from $v$, and we suppose that the rays of $S(v)$ are (circularly) clockwise ordered. By the clockwise range $[ vp,vq]$ of $S(v)$ we mean the ordered set of rays placed clockwise between $vp$ and $vq$, including rays $vp$ and $vq$. When $vp$ or $vq$ or both are not included in that ordered set of rays, we will use $(vp,vq] , [ vp,vq)$ or $(vp,vq) $, respectively. In the same way, we can define counterclockwise ranges.

\begin{figure}[tb]
\centering
\includegraphics[page=1,scale=1.0]{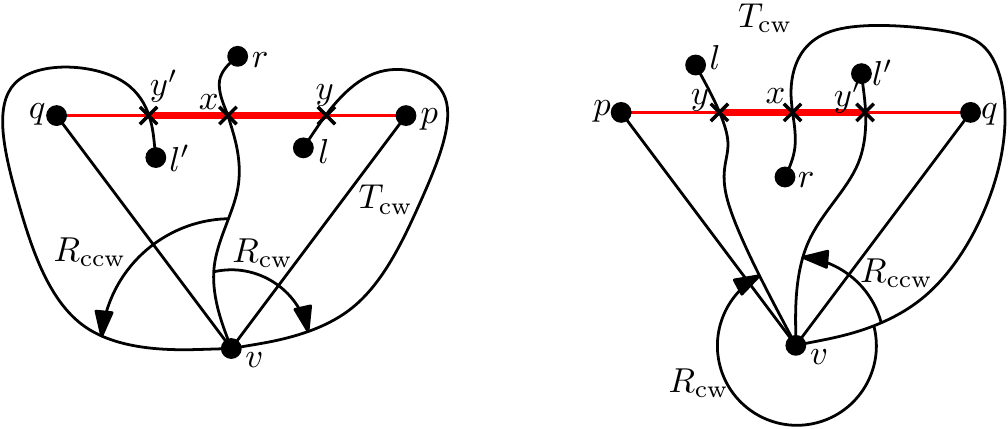}
\caption{The clockwise and counterclockwise ranges of a first crossing.}
\label{Rays1}
\end{figure}

In the rest of this section, we suppose $F$ is a given plane subgraph of $D_n$ and $v$ a vertex not in $F$. In the figures, we use red color for the edges of $F$, so we usually call them \emph{red edges}. We say a ray $vr$ of $S(v)$ is \emph{uncrossed} if it does not cross any edge of $F$.

Suppose that the ray $vr$ crosses some edge of $F$, let $e=pq$ be the first edge of~$F$ crossed by $vr$, and let $x$ be the first crossing point.
Without loss of generality, we can suppose that the rays $vr$, $vp$, and $vq$ appear in this clockwise order in $S(v)$. See \fig{Rays1}.

We define the clockwise range $R_\text{cw}$ of rays centered at $v$ corresponding to the crossing $x$ in the following way:
if no ray in the clockwise range $(vp,vq)$ crosses the edge~$pq$ between $x$ and $p$, then $R_\text{cw}$ is the range $(vr,vp]$;
otherwise, (some rays in the clockwise range $(vp,vq)$ cross the line $xp$), $R_\text{cw}$ is the clockwise range $(vr,vl]$, where $vl$ is the last ray in $(vp,vq)$ crossing the line $xp$. That implies that if $vl$ crosses $xp$ at the point $y$, among the intersection points of rays in $(vp,vq)$  with the line $xp$, the closest to $x$ is $y$.
See \fig{Rays1}.
Analogously, the range $R_\text{ccw}$ is defined either as the counterclockwise range $(vr,vq]$ if no edge in the counterclockwise range $(vq,vp)$ crosses the line $xq$, or as the counterclockwise range $(vr,vl']$, where $vl'$ is the ray in the counterclockwise range $(vq,vp)$ crossing the line $xq$ in a point $y'$ closest to $x$.
By definition, the rays $vr,vp,vl,vl',vq$ appear clockwise in that order around~$v$.
Observe that $R_\text{cw}$ and $R_\text{ccw}$ are disjoint sets and they are also nonempty, as $vp$ is in $R_\text{cw}$ and $vq$ is in $R_\text{ccw}$. The following result generalizes Lemma~\ref{lem:fulek_ruiz_vargas}.\footnote{Like Lemma~\ref{lem:fulek_ruiz_vargas}, this result is more general.
It does not require $F$ and $v$ to be elements of a drawing of $K_n$, but rather of a drawing that contains all edges from $v$ to vertices of $F$.}

\begin{proposition}\label{prop:ReducedRange}
Suppose the ray $vr$ first crosses the edge $e$ of $F$ at the point $x$.
Let $R_\text{cw}$ and $R_\text{ccw}$ be the clockwise and counterclockwise ranges of rays of $v$ corresponding to that crossing.
Then, each one of these two ranges contains an uncrossed ray.
As a consequence, $S(v)$ contains at least two uncrossed rays.
\end{proposition}
\begin{proof}
We prove the statement for $R_\text{cw}$, the proof for $R_\text{ccw}$ is identical.

Observe that by definition, no red edge can cross the line $vx$, and a ray in the clockwise range $(vq,vr]$ cannot cross the line $xp$.
Let $y$ be the crossing point between the red edge $e = pq$ and the ray $vl$. When $R_\text{cw}$ is $(vr,vp]$ (i.e., no ray in the clockwise range $(vp,vq)$ crosses $xp$), then we identify the points $p$, $l$ and $y$.
The lines $vx,xy$ and $yv$ define a simple closed curve, that divides the plane into two regions $T_\text{cw},\overline{T_\text{cw}}$, where $T_\text{cw}$ is the region not containing the point $q$.

From the definition of $T_\text{cw}$, it follows that a ray containing a  point placed in the interior of $T_\text{cw}$ must be in the range $R_\text{cw}$. Besides, a red edge can cross the boundary of that region only through the line $yv$, and hence, if a red edge $e$ crosses $yv$, one endpoint of $e$ must be inside $T_\text{cw}$ the other one in $\overline{T_\text{cw}}$.

The proof is done by induction on $|R_\text{cw}|$, the number of rays in that range.
So, first suppose that the only ray in the range $R_\text{cw}$ is the ray $vp$. In this case, $T_\text{cw}$ is the region bounded by the closed curve $vx,xp,pv$ not containing the point $q$. This region cannot contain any vertex $r'$ of $F$, because then $vr'$ would be in $R_\text{cw}$, therefore $vp$ must be uncrossed.
This proves the base case of the induction.

Now suppose that the proposition has been proved for any clockwise range containing less than $|R_\text{cw}|$ rays.
Let $vr'$ be the first ray of $R_\text{cw}$.
Of course, if the ray $vr'$ is uncrossed, the proof is done, so we can suppose that the ray $vr'$ is first crossed by a red edge $e'$ at a point $x'$.
We are going to prove that the clockwise range $R'_\text{cw}$ corresponding to the crossing $x'$ is strictly contained in $R_\text{cw}$. Then, by induction, $R'_\text{cw}$ contains an uncrossed ray, and thus also~$R_\text{cw}$. To prove that $R'_\text{cw}\subset R_\text{cw}$ strictly, it is enough to prove that the clockwise last ray of $R'_\text{cw}$ is contained in $R_\text{cw}$.

\begin{figure}[!htb]
\centering
\includegraphics[page=2,scale=1.0]{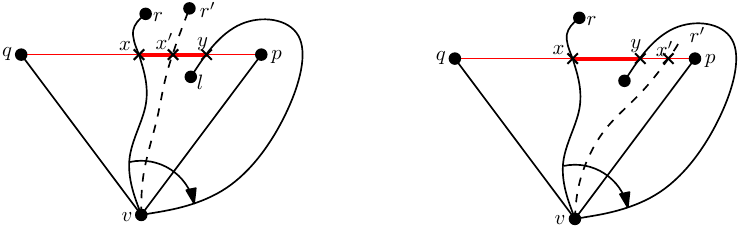}
\caption{Case A, the ray $vr'$ first crosses the edge $e=pq$.}
\label{Rays2}
\end{figure}

Let us first analyze Case A: when the edge $e'$ is precisely the edge $e$. See \fig{Rays2}.
In this case, if $x'$ is between $x$ and $y$, then the clockwise range  $R'_\text{cw}$ corresponding to the new crossing point $x'$ is precisely $R_\text{cw}$ minus its first ray $vr'$. And if $x'$ is between $y$ and $p$, then all the points of the line $x'p$ (including point $p$) are in the interior of the region $T_\text{cw}$, therefore the corresponding last ray of $R'_\text{cw}$ has to be in $R_\text{cw}$. Thus, in both subcases is  $R'_\text{cw}\subset  R_\text{cw}$ strictly.

\begin{figure}[!htb]
\centering
\includegraphics[page=3,scale=1.0]{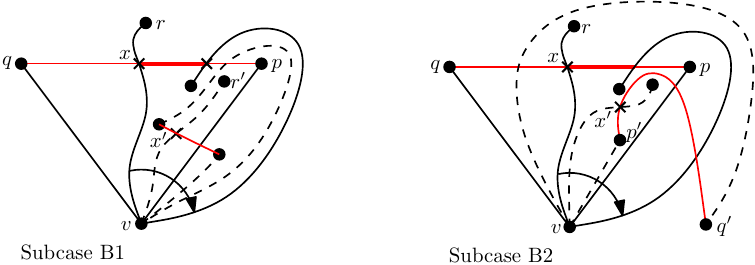}
\caption{Case B, the ray $vr'$ first crosses an edge $e'\neq e$.}
\label{Rays3}
\end{figure}

Suppose now Case B: when $e'\neq e$. See \fig{Rays3}.
Clearly, at least one endpoint of $e'$ is in $T_\text{cw}$, as otherwise the ray $vl$ would be crossed twice by $e'$.
Hence, either both endpoints are in $T_\text{cw}$, subcase B1, or one of them is in $T_\text{cw}$ and the other one in $\overline{T_\text{cw}}$, subcase B2.
In subcase B1, the entire edge $e'=p'q'$ must be inside $T_\text{cw}$. Therefore, any ray containing a point of $e'$ must be in $R_\text{cw}$. In particular, the last ray of $R'_\text{cw}$ must be in $R_\text{cw}$, and hence, $R'_\text{cw}$  is strictly contained in $R_\text{cw}$.

In subcase B2, an endpoint, $p'$, of $e'$ is inside $T_\text{cw}$ and the other, $q'$, is in $\overline{T_\text{cw}}$. Observe that the ray $vp'$ must be in $R_\text{cw}$, however the ray $vq'$ cannot  be in $R_\text{cw}$ because the range $(vr,vr')$ is empty, and a ray in $[vr',vl]$ finishing at $q'$ has to cross the edge $e'$. See \fig{Rays3}.
Therefore, the rays $vr',vp',vq'$ appear clockwise around $v$ in this order. Hence, the last ray of $R'_\text{cw}$ is either $vp'$ or a ray crossing the line $x'p'$. In any case, as the line $x'p'$ is inside $T_\text{cw}$, this last ray of $R'_\text{cw}$ has to be in $R_\text{cw}$.
This completes the proof.
\end{proof}

\section{Structure of maximal plane subgraphs}
Let $D_n$ be an arbitrary simple drawing of $K_n$.
In this section, we identify several structural properties of maximal plane subgraphs of~$D_n$, using Lemma~\ref{lem:fulek_ruiz_vargas} or Proposition~\ref{prop:ReducedRange} as our main tool.
Maximal plane subgraphs turn out to be 2-connected.
While there are examples of maximal plane subgraphs that are not 3-connected, we elaborate further on the structure, showing that a maximal plane subgraph is either 3-edge-connected or has a vertex of degree~2.

\begin{theorem}\label{thm:2-connected}
A maximal plane subgraph of $D_n$ is spanning and 2-connected.
\end{theorem}
\begin{proof}
The proof is by induction on~$n$.
The result is obviously true for $n \leq 3$.
For $n > 3$, assume there exists a maximal plane subgraph $\overline{F}$ that is not 2-connected, and let us see that a contradiction is reached.

We first claim that, under this assumption, $\overline{F}$ has no vertices of degree less than~3.
Suppose the contrary, that the vertex $v$ has degree $\le 2$. Let $F'$ be the subgraph of $\overline{F}$ obtained after removing the vertex $v$, and let $\overline{F'}$ be a maximal plane subgraph (in the drawing $D_n - \{v\}$ of $K_{n-1}$) containing~$F'$.
By the induction hypothesis, $\overline{F'}$ is 2-connected.
We observe that $v$ cannot have (in $\overline{F}$) degree less than~2, since applying Lemma~\ref{lem:fulek_ruiz_vargas} to $v$ and $\overline{F'}$ would give two edges at $v$ not crossing $\overline{F}$, contradicting the maximality of~$\overline{F}$.
So suppose $v$ has degree~2.
As we assume that $\overline{F}$ is not 2-connected, $F'$ cannot be 2-connected.
However, $\overline{F'}$ is 2-connected, and hence there exists an edge $e'$ in $\overline{F'} - F'$.
By the maximality of $\overline{F}$, $e'$ must cross at least one edge $vw$ of $\overline{F}$ incident to $v$.
But applying Lemma~\ref{lem:fulek_ruiz_vargas} to $v$ and $\overline{F'}$ gives at least two edges incident to $v$ not crossing $\overline{F'}$.
These two edges and also $vw$ do not cross $\overline{F}$, contradicting the maximality of $\overline{F}$. Therefore, the claim follows.



Assume now that $\overline{F}$ is not connected.
Let $C_1,C_2$ be two connected components of $\overline{F}$.
As all vertices have (in $\overline{F}$) degree at least 3, $C_1$ cannot be an outerplanar graph, and it has more than one face.
Without loss of generality, we can suppose that $C_2$ is in the unbounded face of $C_1$.
Let $v_1$ be an interior vertex of~$C_1$, $F'$  the graph obtained from $\overline{F}$ by removing $v_1$, and $f_1$ the face of $F'$ containing~$v_1$.
The face containing $C_2$ remains unchanged by the removal of $v_1$.
By induction, $F'$ can be completed to a 2-connected plane graph $\overline{F'}$, and
due to the maximality of $\overline{F}$, all the edges in $\overline{F'} - F'$ should be in the face $f_1$. But then,
as $C_2$ is outside~$f_1$, $\overline{F'}$ could not be connected, a contradiction.
Thus, $\overline{F}$ has to be connected.

By a similar reasoning we arrive at our contradiction to $\overline{F}$ not being 2-connected.
A \emph{block} is a 2-connected component of a graph, and a \emph{leaf block} is a block with only one cut vertex.
Since $\overline{F}$ is not 2-connected, it has at least two leaf blocks $B_1$ and $B_2$.
As all vertices have degree at least 3, $B_1$ cannot have all its vertices on the same face.
Again, without loss of generality, we can suppose $B_2$ is in the outer face of $B_1$, and there is an interior vertex $v_1$ of $B_1$.
Removing $v_1$ from $\overline{F}$, we obtain a plane graph $F'$ that has a face $f_1$ containing $v_1$, and $F'$ is contained in a maximal plane graph $\overline{F'}$ that is 2-connected.
Again, by the maximality of $\overline{F}$, all the edges in $\overline{F'} - F'$ must be in $f_1$, implying that $B_2$ is still a block of $\overline{F'}$, contradicting the fact that $\overline{F'}$ is 2-connected.
Hence, $\overline{F}$ must be 2-connected.
\end{proof}

Theorem~\ref{thm:2-connected} can be used to obtain more properties of maximal plane subgraphs.

\begin{lemma}\label{lem:no_two_consecutive}
If a maximal plane subgraph~$\overline{F}$ of $D_n$ contains a vertex $v$ of degree~2, then the subgraph of $\overline{F}$ obtained after removing $v$ is also maximal in $D_n - \{v\}$. 
\end{lemma}
\begin{proof}
Suppose the contrary.
Remove $v$ from $\overline{F}$ to obtain $F'$ and let $\overline{F'}$ be a maximal plane graph containing $F'$.
As $\overline{F}$ is maximal but $F'$ is not, $\overline{F'}-F'$ must contain an edge $e'$ that crosses some edge $vw$ of $\overline{F}$.
But by Lemma~\ref{lem:fulek_ruiz_vargas} there are at least two edges from $v$ to $\overline{F'}$. These two edges and also $vw$ do not cross $\overline{F}$, contradicting the maximality of $\overline{F}$.
\end{proof}

\begin{proposition}\label{prop:edge_number}
Any maximal plane subgraph $\overline{F}$ of $D_n$ with $n\geq 3$ must contain at least $\min(\ceil{3n/2}, 2n-3)$ edges.
This bound is tight.
\end{proposition}
\begin{proof}
Suppose that $n > 3$ and $\overline{F_0}=\overline{F}$ has a vertex $v_0$ with degree 2.
By removing this vertex we obtain another maximal plane graph $\overline{F_1}$ (maximal on $n-1$ points), and if $\overline{F_1}$ is in the same conditions (with at least three vertices and a vertex $v_1$ of degree 2), by removing $v_1$ we obtain a new maximal plane graph $\overline{F_2}$, and so on.
We finish this process in a step $k$ because either $\overline{F_k}$ only has three points, or all the points of $\overline{F_k}$ have degree at least 3.
In the first case, the original graph $\overline{F}$ contains $n=k+3$ vertices and $2k+3$ edges, so $2n-3$ edges.
In the second case, $\overline{F}$ must contain at least $2k+\lceil 3(n-k)/2\rceil$ edges, this amount reaching its minimum value when $k=0$.

Finally, let us see that the bound is tight.
If $2\le n\le6$, then a straight-line drawing on $n$ points in convex position gives the bound $2n-3\le \lceil 3n/2\rceil $.
If $n>6$ and $n$ is an even number, a drawing like the one shown in \fig{TightBound} proves that the bound $\lceil 3n/2\rceil $ is tight.
The drawing is done on $n=2(k+1)$ points in convex position, that clockwise are denoted by $u_0,u_1,u_2,\ldots ,u_k,v_{k+1},v_k,\ldots ,v_2,v_1$. Let $C$ denote the convex hull of that set of points.
All the edges of $D_n$ are drawn straight-line except for the $2(k-1)$ edges $u_iv_{i+1},v_iu_{i+1}, i=1,\ldots ,k-1$, and the edge $u_0v_{k+1}$, that are drawn outside $C$ as shown in \fig{TightBound}.
Observe that the $2(k-1)$ edges $u_iv_{i+1},v_iu_{i+1}, i=1,\ldots ,k-1$, are the diagonals of the $(k-1)$ quadrilaterals $u_iu_{i+1}v_{i+1}v_i$, with $u_iv_{i+1}$ only crossing $v_iu_{i+1}$ and $u_0v_{k+1}$, for $i=1,\ldots ,k-1$.
Clearly, straight-line edges can cross at most once, and the edges placed outside $C$, by construction, cross at most once.
The graph $\overline{F}$ formed by the $2(k+1)$ edges on the boundary of $C$, the $k$ edges $u_iv_i, i=1,\ldots ,k$, and the edge $u_0v_{k+1}$ is clearly plane and maximal, since the other straight-line edges cross at least one edge $u_iv_i$, and the non-straight-line edges cross the edge $u_0v_{k+1}$.

If $n$ is odd, we can add to the previous set a point $u_0'$ between $u_0$ and $u_1$, very close to segment $u_0u_1$, but keeping all the $2k+3$ points in convex position.
By connecting $u'_0$ with straight lines to the rest of the points, we obtain a simple topological drawing of $K_n$ on this set of $n=2k+3$ points, and a new maximal plane graph is obtained by adding the edges $u_0u'_0$,$u'_0u_1$ to the above graph $\overline{F}$.
This new maximal plane subgraph also has $\lceil 3n/2\rceil $ edges.
\end{proof}

\begin{figure}[!tb]
\centering
\includegraphics[page=4,scale=0.9]{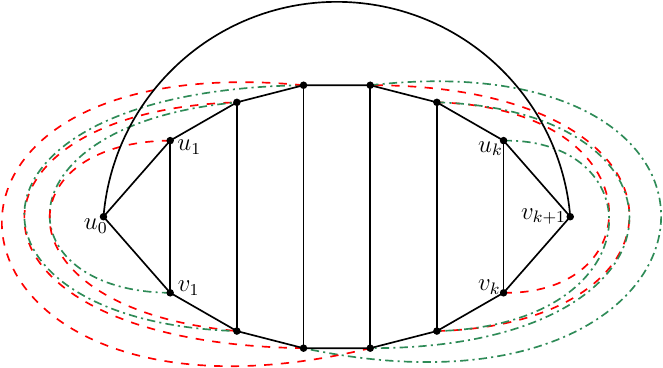}
\caption{A drawing of $K_n$. The missing edges should be drawn as straight-line segments inside the convex hull of the set of points. The black edges form a maximal plane subgraph with $\lceil 3n/2\rceil $ edges.
}
\label{TightBound}
\end{figure}

We mention another interesting implication of Theorem~\ref{thm:2-connected}. For a vertex~$v$, we can augment the star $S(v)$ to a 2-connected plane graph $\overline{F}$,
and since $\overline{F}\setminus \{v\}$ is connected, it contains a spanning tree. So we have
\begin{corollary}\label{thm:star_tree}
For each vertex $v$ there exists a spanning tree $T_v$ of $V \setminus \{v\}$, such that the edges of $S(v)\cup T_v$ form a plane subgraph of~$D_n$.
\end{corollary}

Our next results are about diagonals on plane cycles. Let $C= (v_1, v_2, \dots, v_k)$ be a plane cycle of $D_n$. A diagonal of $C$ is an edge of $D_n$ connecting two non-consecutive vertices of $C$. It was previously known that, even for the case where there are diagonals intersecting both faces of $C$, there are at least $\ceil{k/3}$ of them not crossing~$C$ (cf.~\cite[Corollary~6.6]{pammer}).
Proposition~\ref{prop:edge_number}, applied to the subdrawing induced by the vertices of $C$, directly implies the following result.

\begin{corollary}
Let $C= (v_1, v_2, \dots, v_k)$ be a plane cycle of $D_n$, with $k \geq 6$.
Then, there exists a set $D$ of diagonals, with $|D|\geq \ceil{k/2}$, such that the subgraph $C\cup D$ is plane.
\end{corollary}

It turns out that the structure of the diagonals of a cycle, as shown in the next lemma, is useful for our further results.

\begin{lemma}\label{lem:empty_faces}
Let $C= (v_1, v_2, \dots, v_k)$ be a plane cycle of $D_n$, $k \geq 3$, dividing the plane into two faces $f_1$ and $f_2$.
If there is no diagonal of $C$ entirely in $f_1$, then all the diagonals of $C$ are entirely in~$f_2$.
\end{lemma}
\begin{proof} 
The proof is by induction on $k$.
For $k<5$ the statement is obvious, so suppose $k\ge 5$ and consider only the subdrawing $D_k$  induced by the vertices of $C$. Suppose $C\cup D$ is a maximal plane graph of $D_k$, so necessarily,  $D$ consists of diagonals placed on $f_2$. Let $d$ be a diagonal of $D$ connecting two vertices at minimum distance on the graph $C$.
Lemma~\ref{lem:no_two_consecutive} implies that in a maximal plane subgraph, vertices with degree 2 cannot be adjacent. Therefore, diagonal $d$ has to connect two vertices at distance 2 on $C$.
Without loss of generality, suppose $d=v_k v_2$ and let $\Delta$ be the triangle $v_k v_1 v_2$.
Then, the cycle $C_1=(v_2,v_3,\ldots,v_k)$ with $k-1$ vertices has the faces $f'_1 = f_1 + \Delta$ and $f'_2=f_2-\Delta$.

We claim that there cannot be diagonals of $C_1$ entirely in $f'_1$.
Such a diagonal $e$ entirely in $f'_1$ would have to intersect $\Delta$.
Then, adding $e$ to $C \cup \{v_k v_2\}$ and removing all edges crossed by $e$, we would obtain a plane graph $F$ in which $v_1$ has degree 0 or 1.
By Lemma~\ref{lem:fulek_ruiz_vargas}, there must be another edge between $v_1$ and $C_1$, and this edge would be a diagonal of $C$ entirely in~$f_1$, a contradiction.
Thus, by induction, any diagonal $v_iv_j$ of $C_1$ is entirely in $f'_2$ and hence also in $f_2$.

It remains to see that the diagonals with endpoint $v_1$ are also in $f_2$.
By our induction hypothesis, the diagonal $v_2v_4$ is in $f'_2$ and thus also in $f_2$.
Hence, arguing as before on the cycle $C_3=(v_1,v_2,v_4,\ldots ,v_k)$, we deduce that all the diagonals of $C_3$ incident to $v_1$ must be in $f_2$.
So it remains to see that the diagonal $v_1v_3$ is also in $f_2$.
But $v_3v_5$ is also in~$f'_2$, so it is in $f_2$, and again applying the same reasoning on the cycle $(v_1,v_2,v_3,v_5,\dots,v_k)$, all the diagonals of this cycle not incident to $v_4$ have to be in~$f_2$.
\end{proof}

To prove the next result, we recall some definitions and properties of any 2-connected graph $G=(V,E)$.
Two vertices $v_1,v_2$ are called a \emph{separation pair} of $G$ if the induced subgraph $G \setminus \{v_1,v_2\}$ on the vertices $V \setminus \{ v_1, v_2\}$ is not connected.
Let $G_1,\ldots ,G_l$ be the connected components of $G \setminus \{ v_1, v_2\} $, with $l\ge 2$.
For each $i\in\{1,\ldots ,l\}$, let $G^*_i$ be the subgraph of $G$ induced by  $V(G_i)\cup \{ v_1, v_2\}$.
Observe that $G^*_i$ contains at least one edge incident to $v_1$ and at least another edge incident to $v_2$.

\begin{theorem}\label{thm:2-connected_Component}
Let $\overline{F}$ be a maximal plane subgraph of $D_n$, $n \geq 3$.
Then, for each separation pair $v_1,v_2$ of $\overline{F}$, at least one of the subgraphs $\overline{F}^*_i$ must be 2-connected.
\end{theorem}
\begin{proof}

Suppose that $v_1,v_2$ is a separation pair of $\overline{F}$, and let  $\overline{F}_1,\overline{F}_2,\ldots ,\overline{F}_l$ be the connected components of $\overline{F}\setminus \{ v_1,v_2\} $, $l\ge 2$. Since $\overline{F}$ is 2-connected, the graph $\overline{F}\setminus \{v_2\}$ is connected with $v_1$ as a cut vertex.
As $\overline{F}$ is plane, we can suppose that $v_1$ is in the outer face of $\overline{F}\setminus \{v_2\}$ ($v_2$ must be inside that face) and that clockwise around vertex $v_1$ first there appear the edges from $v_1$ to some vertices of the component $\overline{F}_1$, then edges connecting $v_1$ to some vertices of $\overline{F}_2$ and so on.
See \fig{2_connected_components}.

\begin{figure}[!htb]
\centering
\includegraphics[page=5,scale=0.9]{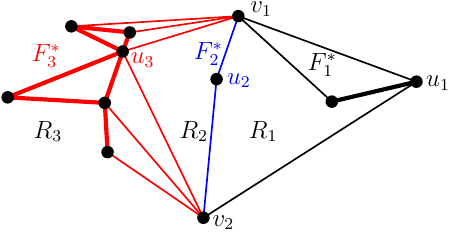}
\caption{A plane graph with separating pair $v_1, v_2$ and three subgraphs $F^*_i$, none of them 2-connected. This plane graph cannot be maximal.}
\label{2_connected_components}
\end{figure}

Now suppose that none of the subgraphs $\overline{F}^*_i$ is 2-connected.
Then each subgraph $\overline{F}^*_i$ contains at least one cut vertex~$u_i$.
Since $\overline{F}_i$  is connected and there exist edges in $\overline{F}^*_i$ incident to $v_1$ and $v_2$, vertex $u_i$ is different from $v_1$ and $v_2$.
On the other hand, a connected component $C$ of $\overline{F}^*_i\setminus \{u_i\}$ must contain at least one of $v_1$ or $v_2$ because otherwise, $C$ would be a connected component of $\overline{F}\setminus \{u_i\}$, contradicting that $\overline{F}$ is 2-connected. Therefore, $\overline{F}^*_i\setminus \{u_i\}$ has exactly two components, one containing $v_1$, the other one containing $v_2$. This also implies that the edge $v_1v_2$ of $D_n$ cannot belong to $\overline{F}$, and that the cut-vertex $u_i$ is in the outer face of $\overline{F}^*_i$ (and hence in the outer face of $\overline{F}\setminus \{v_2\}$) since $v_1$ and $v_2$ are in the outer face of $\overline{F}\setminus \{v_2\}$. See \fig{2_connected_components}.

In the graph $\overline{F}\setminus \{v_2\}$, around the vertex $v_1$, the edges to vertices of $\overline{F_1}$ first appear, then the edges to vertices of $\overline{F_2}$ and so on. Therefore, when we add to that graph the vertex $v_2$ and all the edges connecting $v_2$ to each component $\overline{F_i}$ to obtain $\overline{F}$,  $v_1$ and $v_2$ must be in the faces $R_i$ of $\overline{F}$ defined as the regions placed between the last edge from $v_1$ to $\overline{F}_i$ and the first edge from $v_1$ to $\overline{F}_{i+1}$, for $i=1,\ldots ,l$, and the vertex $u_i$ must be in the faces $R_i$ and $R_{i-1}$.
However, by the maximality of $\overline{F}$, no edge of $D_n$ is entirely in any of those faces $R_i$. Then, Lemma~\ref{lem:empty_faces} implies that no point of the edge $v_1v_2$ of $D_n$ can be inside any face $R_i$. See \fig{2_connected_components}.
Thus, $v_1v_2$ must begin between two edges $v_1v,v_1v'$ with both $v$ and $v'$ belonging to a common connected component $\overline{F}_i$.
However, since $u_i$ belongs to the faces $R_{i-1}$ and $R_i$, any curve from $v_1$ to $v_2$ passes either through the point $u_i$ or through the interior of $R_{i-1}$ or $R_{i}$, which contradicts either the simplicity of $D_n$ or Lemma~\ref{lem:empty_faces}. Therefore, if none of the subgraphs $\overline{F}^*_i$ is 2-connected, $\overline{F}$ cannot be maximal.
\end{proof}

We call a graph \emph{essentially 3-edge-connected} if it stays connected after removing any two edges not sharing a vertex of degree~2 (i.e., the graph either stays connected or one component is a single vertex). Theorem~\ref{thm:2-connected_Component} implies that a maximal plane subgraph is essentially 3-edge-connected:
\begin{theorem}\label{thm:essentially}
Any maximal plane subgraph, $\overline{F}$, of a simple topological drawing of $K_n$ is essentially 3-edge-connected.
\end{theorem}
\begin{proof}

If the removal of two edges $v_1 v_2$ and $v_1' v_2'$ from the plane subgraph $\overline{F}$ results in two non-trivial components $C_1,C_2$ (see \fig{3-edge-connectedF}), then $v_1,v'_2$ is a separation pair of $\overline{F}$, that has as induced subgraphs $C_1 \cup \{v'_1 v'_2\}$ and $C_2 \cup \{v_1 v_2\}$, neither of which is 2-connected. Then, by
 Theorem~\ref{thm:2-connected_Component}, $\overline{F}$ cannot be maximal.
\end{proof}
\begin{figure}[!htb]
\centering
\includegraphics[page=6,scale=0.9]{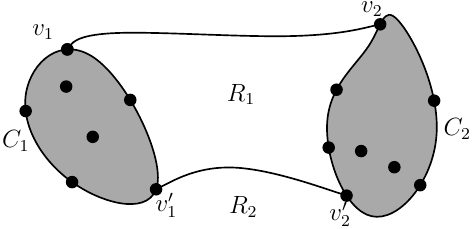}
\caption{A graph that is not essentially 3-edge-connected.
The induced subgraphs of the separation pair $v_1, v'_2$ are subgraph $C_1$ plus edge $v'_1 v'_2$ and subgraph $C_2$ plus edge $v_1 v_2$. By Lemma~\ref{lem:empty_faces}, the edge $v_1v'_2$ of $D_n$ cannot enter either the $R_1$ face or the $R_2$ face, which is impossible in any good drawing.}
\label{3-edge-connectedF}
\end{figure}


\section{Adding the maximum number of edges}
Now, assume that a plane subgraph of $D_n$ is given, and we want to add the maximum number of edges keeping plane the augmented graph. Clearly, the decision of adding one edge will in general block other edges from being added. We will see that the complexity of an algorithm solving this problem highly depends on whether the given subgraph is connected or not.

Before talking about algorithms and their complexity we have to talk about what information of the drawing $D_n$ we will need to compute plane subgraphs. For each vertex $v$, the clockwise cyclic order of edges of $S(v)$ is usually given as a permutation of $V\setminus \{ v\} $ (that is to be interpreted circularly) of the second vertices of all edges of $S(v)$. That permutation of $V\setminus \{ v\} $ is called the \emph{rotation} of $v$, and the \emph{rotation system} of a drawing $D_n$ consists of the collection of the rotations of each vertex $v$ of $D_n$. It is well-known that from the information provided by the rotation system, one can determine whether two edges cross or not, and therefore, that information is enough to compute plane subgraphs. See~\cite{gioan, pach_toth_2000, realizability_in_p}. From the rotation system, we can also compute (in $O(n^2)$ time) the \emph{inverse rotation system} that, for each vertex~$v_i$ and index $j$, $j \neq i$, gives the position of $v_j$ in the rotation of~$v_i$.

When we say that a drawing $D_n$ is given, we mean that we know the rotation system and the inverse rotation system of $D_n$. Using these two structures, one can determine whether two edges cross, in which direction an edge is crossed, and in which order two non-crossing edges cross a third one in constant time~\cite{realizability_in_p}.

\begin{theorem}
Let $F$ be a connected spanning plane subgraph of $D_n$.
Then there is an $O(n^3)$ time algorithm to augment $F$ to a plane subgraph $F'$ of $D_n$ with the maximum number of edges.
\end{theorem}
\begin{proof}
As $F$ is plane and thus contains a linear number of edges, we can identify all the edges of $D_n$ not crossed by $F$ in $O(n^3)$ time.
This also gives us, for each such edge, the face of $F$ in which it is contained, and we can also compute for each face $f$ of $F$ the set $\Delta_f$ of edges of $D_n$ entirely inside~$f$.
Clearly, each face of $F$ can be considered independently, adding the maximum number of edges in it.

Let $f$ be a face of $F$. For simplicity, we assume $f$ to be bounded by a simple cycle $(v_1, \dots, v_k)$. Other cases can be solved
similarly by an appropriate ``splitting'' of edges having $f$ on both sides.
Disregarding $D_n$, consider the rectilinear drawing $\overline{D_k}$ obtained from $k$ points $p_1,\ldots ,p_k$ placed on a circle $C$, and assign to each edge $p_ip_j$  of $\overline{D_k}$ weight $0$ if $v_iv_j$ is in $\Delta_f$, weight $1$ otherwise. Observe that two edges of $\Delta _f$ cross properly, if and only if, the corresponding 0-weight edges in circle $C$ cross properly. It is well-known that a minimum-weight triangulation in $\overline{D_k}$ can be obtained in $O(k^3)$ time~\cite{klincsek} by a dynamic programming algorithm, and this triangulation gives a plane set of 0-weight edges with maximum cardinality. Hence, the corresponding edges of $\Delta _f$ form a plane set of edges entirely inside face $f$ with maximum cardinality.
\end{proof}

In contrast to this result, the problem becomes NP-complete when the subgraph $F$ is not connected.

\begin{theorem}\label{thm:augmenting_hard}
Given a simple topological drawing $D_n$ of $K_n$ and a cardinality $k'$, it is NP-complete to decide whether there is a plane subgraph that has at least $k'$ edges.
\end{theorem}
%
%
\begin{proof}
We give a reduction from the independent set problem on segment intersection graphs ($\overline{SEG}$ problem), which is known to be NP-complete~\cite{segment_independent_set}:
Given a set $S$ of $s$ segments in the plane that pairwise either are disjoint or intersect in a proper crossing, and an integer $k>0$, is there a subset of $k$ disjoint segments?

\begin{figure}[tb]
\includegraphics[page=7,scale=0.86]{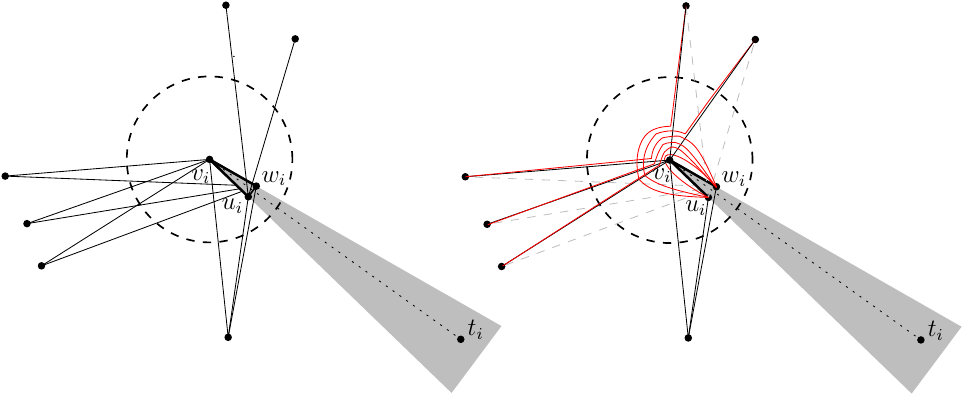}
\caption{Drawings $\overline{D_n}$ (left) and $D_n$ (right).
The gray wedge only contains the endpoint $t_i$.
In $D_n$, the dashed edges need to take a detour to avoid intersecting the edge $u_iw_i$ twice.}
\label{fig:pit}
\end{figure}

For each instance of a $\overline{SEG}$ problem, we are going to build, in polynomial time, a drawing $D_n$ of $K_n$ and an integer $k'$ such that the instance of the $\overline{SEG}$ problem has a Yes answer, if and only if, the drawing $D_n$ contains a plane subgraph with $k'$ edges.

Let $v_i,t_i$ be the endpoints of each segment~$s_i,i=1,\ldots ,s,$ of $S$. We can suppose that these endpoints are in general position and that their convex hull is a triangle. Thus, for each endpoint $v_i$, we can find a disc $B_i$ centered at $v_i$, such that any straight line connecting two endpoints of $S$ different from $v_i$ does not cross $B_i$.

In each disc $B_i$, we place two points $u_i,w_i$ very close to the segment $v_it_i$, in such a way that when connecting the point $v_i$ with straight-line segments to all the
other points, the segments $v_iu_i,v_it_i,v_iw_i$ are clockwise consecutive. In other words, the clockwise wedge defined by the half-lines $v_iu_i,v_iw_i$ only contains the endpoint $t_i$. See \fig{fig:pit}.

Consider the rectilinear drawing $\overline{D_n}$ obtained by connecting the $n=4s$ points $v_i,u_i$, $w_i,t_i$. In $\overline{D_n}$, maximal plane graphs are triangulations, but we are going to consider only the family $\Gamma $ of plane triangulations of $\overline{D_n}$ containing the $2s$ edges $u_iv_i,w_iv_i$. The \emph{weight} of a triangulation of $\Gamma $ is the number of edges $v_it_i$ that it contains. Clearly, in the set $S$ there are $k$ disjoint segments, if and only if, there is a triangulation in $\Gamma $ with weight $k$.

Now, consider the drawing $D_n$ obtained from $\overline{D_n}$ doing the following changes:

For $i=1,\ldots ,s$, only the edges of the star $S(u_i)$ crossing  $v_iw_i$, the edges of the star $S(w_i)$ crossing $u_iv_i$, and the edge $u_iw_i$ are modified.

Suppose that in $S(u_i)$ after $u_iv_i$ are clockwise the edges $u_ip_1,\ldots ,u_ip_k,u_iw_i$,  where each $u_ip_j$ has to cross $v_iw_i$. Let $v_ip_{i_1},v_ip_{i_2}\ldots ,v_ip_{i_k}$  be the clockwise ordered edges of $S(v_i)$ with endpoint one of the vertices $p_i$. Then, we modify $\overline{D_n}$ by redrawing $u_ip_{i_1}$ following first the line $u_iv_i$ until point $v_i$, then turning around $v_i$ and following the line $v_ip_{i_1}$, in such a way that in the rotation of $u_i$ the new edge $u_ip_{i_1}$ is placed just before $u_iv_i$. See \fig{fig:pit}, right. The new drawing obtained is simple, because no edge crosses both $u_iv_i$ and $v_ip_{i_1}$, edges $u_ip_j$ cannot cross $v_ip_{i_1}$ and none edge of $S(p_{i_1})$ can cross $u_iv_i$. Moreover, the number of crossings in the edge $v_iw_i$ has decreased by one. We repeat the same process for the edge $u_ip_{i_2}$ (the new edge $u_ip_{i_2}$ is placed just before $u_iv_i$ in the rotation of $u_i$ ), then $u_ip_{i_3}$, and so on. The same process can be done with the edges $w_iq_j$ crossing $u_iv_i$. See \fig{fig:pit}, right. Finally, we can redraw $u_iw_i$ in the same way, following the edge $u_iv_i$ then turning around $v_i$ following edge $v_iw_i$. If we do this process for all the edges crossing $v_iu_i$ or $u_iw_i$, $i=1,\ldots ,s$, at the end we obtain the simple drawing $D_n$. By construction, in $D_n$, neither the edges $v_iu_i$ nor the
edges $v_iw_i$ are crossed by any other edge.

Now, let us see that $\overline{D_n}$ has a triangulation of the family $\Gamma $ of weight $k$, if and only if, $D_n$ has a plane subgraph of size $k'$, with $k'=3n-6-(s-k)=11s-6+k$. Suppose $\overline{D_n}$ has a triangulation $F$ with weight $k$. This means that $F$ contains $(s-k)$ edges $u_iw_i$. By removing from $F$ these $u_iw_i$ edges, we obtain a plane set $F'$ of edges, where no edge of $F'$ has been modified to obtain the drawing $D_n$. Therefore, the edges of $F'$ also form a plane subgraph in $D_n$ of size $3n-6-(s-k)=11s-6+k$.

Conversely, suppose $D_n$ contains a plane subgraph with $3n-6-(s-k)$ edges. Since the edges $u_iv_i,w_iv_i$ are not crossed by any edge of $D_n$, they must belong to any maximal plane graph of $D_n$. Therefore, $D_n$ has a plane subgraph ${F}$ containing all the edges $u_iv_i,v_iw_i$ and of size $k'\ge 3n-6-(s-k)$. As the wedge $v_iw_i,v_iu_i$ only contains point $t_i$, if the edge $v_it_i$ is not in ${F}$, then, the face of $F$ containing the edges $v_iu_i$ and $v_iw_i$  cannot be a triangle. But, if a plane graph on $n$ vertices contains more than $(s-k)$ non-triangular faces, its maximum number of edges is $<3n-6-(s-k)$. As a consequence, $v_it_i$ is not in ${F}$ for at most $(s-k)$ indices $i$, or equivalently, the plane subgraph ${F}$ contains at least $k$ edges $v_it_i$. This means that we can obtain a triangulation of the family $\Gamma $ of weight $k$ by including
$k$ of these non-crossing edges.
\end{proof}

Note that in the straight-line setting, we can always draw a triangulation of the underlying point set, which contains the maximum number of edges.
However, this is not the case for simple topological drawings.
We were not able to come up with a reduction solving the following problem.

\begin{problem}\label{problem:good_triangulation}
What is the complexity of deciding whether a given $D_n$ contains a triangulation, i.e., a plane subgraph whose faces are all 3-cycles?
\end{problem}

Our reduction can also be adapted for a related problem on compatible graphs.
We leave the realm of general simple topological drawings and consider the following problem in the more specialized setting of geometric graphs (rectilinear drawings).
Let $P = \{p_1, \dots, p_n\}$ and $P' = \{p_1', \dots, p_n'\}$ be two sets of points in the plane.
A planar graph is \emph{compatible} if it can be embedded on both $P$ and $P'$ in a way that there is an edge $p_i p_j$ if and only if there is an edge $p_i' p_j'$.
Saalfeld~\cite{saalfeld} asked for the complexity of deciding whether two such point sets (with a given bijection between them) have a compatible triangulation. We will say that triangulations $\overline{F}$ of $P$ and $\overline{F'}$ of $P'$ have $k'$ compatible edges when there exists a subset of $k'$ edges $p_ip_j$ of $\overline{F}$, such that their images, edges $p'_ip'_j$, are edges of $\overline{F'}$.

We can show the NP-completeness of the following optimization variant of the problem.
(However, as the similar Open Problem~\ref{problem:good_triangulation}, Saalfeld's problem remains unsolved.)

\begin{theorem}
Given two point sets $P = \{p_1, \dots, p_n\}$ and $P' = \{p_1', \dots, p_n'\}$ and the indicated bijection between them, as well as a cardinality $k'$, the problem of deciding whether $P$ and $P'$ admit two triangulations  with $k'$ compatible edges is NP-Complete.
\end{theorem}
\begin{proof}

We follow the idea of the proof of Theorem~\ref{thm:augmenting_hard}, and use a reduction from the $\overline{SEG}$ problem. Suppose that an instance of the $\overline{SEG}$ problem is given: a set $S$ of $s$ segments in the plane that pairwise either are disjoint or intersect in a proper crossing, and an integer $k>0$. We will build two sets of points $P = \{p_1, \dots, p_n\}$ and $P' = \{p_1', \dots, p_n'\}$ and obtain an integer $k'$ such that the $\overline{SEG}$ problem has answer Yes if and only if, $P$ and $P'$ admit two triangulations with $k'$ compatible edges.

Let $P$ be the set of $n=5s$ points formed by the $v_i,t_i,u_i,w_i (i=1,\ldots ,s)$ points obtained from $S$ as in the above Theorem~\ref{thm:augmenting_hard}, plus $s$ points $\tilde v_i$, where each point $\tilde v_i$ is placed inside the triangle $v_iu_iw_i$ very close to the point $v_i$, to the right of the oriented line $v_it_i$, in such a way that in the wedge defined by the half-lines $\tilde v_iu_i,\tilde v_iw_i$ the only point of $P$ is $t_i$, and the wedges $u_iv_i,u_i\tilde v_i$ and $w_i\tilde v_i,w_iv_i$ do not contain points of $P$. See \fig{fig:compatible_pit}, left. By construction, any triangulation ${\overline{F}}$ of the set of points $P$ must contain the edge $v_i\tilde v_i$. Observe that if the edge $t_iv_i$ is in ${\overline{F}}$, then the edge $t_i\tilde v_i$ has to be also in ${\overline{F}}$. Also note that $u_iw_i$ is only crossed by the edges $t_iv_i$ and $t_i\tilde v_i$.

In the same way, let $P'$ be the set of $n=5s$ points  $v_i,t_i,u_i,w_i,\tilde v'_i (i=1,\ldots ,s)$, where now each point $\tilde v'_i$ is placed outside the triangle $v_iu_iw_i$, very close to the intersection point of $u_iw_i$ with $v_it_i$, to the right of the line $v_it_i$, and satisfying that any clockwise triangle $u_iw_ip$ contains inside the point $\tilde v'_i$. See \fig{fig:compatible_pit}, right.
The bijection between the points of $P$ and $P'$ is the obvious one,  to each  point $\tilde v_i$ of $P$ corresponds point $\tilde v'_i$ of $P'$, for any other point its image is itself.

\begin{figure}
\centering
\includegraphics[page=8,scale=0.86]{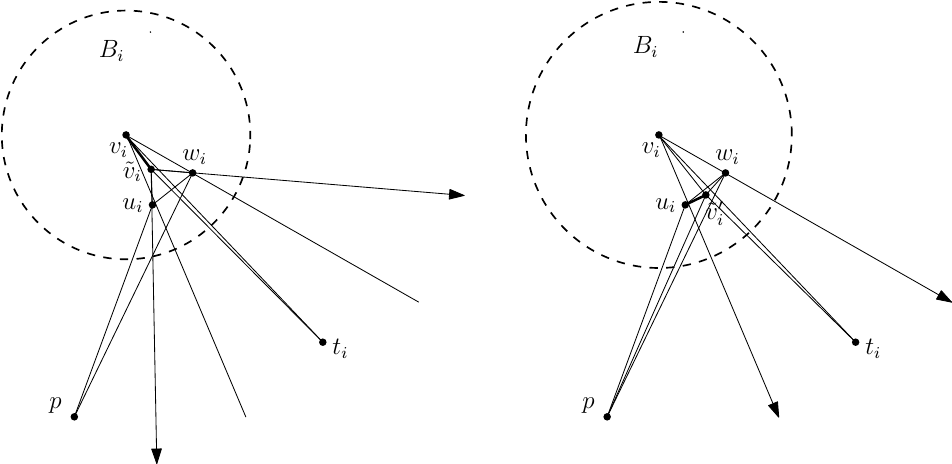}
\caption{The point sets $P$ (left) and $P'$ (right).}
\label{fig:compatible_pit}
\end{figure}

To prove the statement of the theorem, it is enough to prove the following:

If in the set $S$ there are $k$ disjoint segments, then there are triangulations ${\overline{F}}$ and $\overline{F'}$ of the sets $P$ and $P'$, respectively, with  $k'=3n-6-(s-k)$ compatible edges. And reciprocally, if ${\overline{F}}$ and $\overline{F'}$ contain  $k'=3n-6-(s-k)$ compatible edges, then $S$ contains $k$ disjoint segments.

Suppose first that $S$ contains a set $D$ of $k$ disjoint segments $v_it_i$. Let $P_0$ be the set of $4s$ common points of $P$ and $P'$ (all the points $v_i,u_i,w_i,t_i$). We build a triangulation $\overline{F_0}$ of $P_0$ in the following way. If $v_it_i$ is in $D$, then we include the edges $v_it_i,v_iu_i,v_iw_i,t_iu_i,t_iw_i$ in $\overline{F_0}$. If $v_it_i$
is not in $D$, then we include the edges $u_iv_i,v_iw_i,w_iu_i$ in $\overline{F_0}$. After that, we add edges in an arbitrary way until obtaining a triangulation $\overline{F_0}$
of $P_0$. Now, to obtain $\overline{F}$ and $\overline{F'}$, we add the points $\tilde v_i$ and $\tilde v'_i$ to $\overline{F_0}$ and retriangulate the triangular faces where they are.
If the edge $v_it_i$ is in $D$, then the points $\tilde v_i,\tilde v'_i$ are both in the triangle $u_iv_it_i$. So, by adding  the  point $\tilde v_i$ and the three edges $\tilde v_iu_i,\tilde v_iv_i, \tilde v_i t_i$ to $\overline{F_0}$, or the point $\tilde v'_i$ and the three edges $\tilde v'_iu_i,\tilde v'_iv_i, \tilde v'_i t_i$ we continue with all the edges being compatible. 
However, if the edge $v_it_i$ is not in $D$, then the point $\tilde v_i$ is in the triangle $u_iv_iw_i$, but the point $\tilde v'_i$ is in a triangle $u_iw_ip_i$. Then, we obtain a triangulation $\overline{F}$ of $P$ by adding the edges $\tilde v_iu_i,\tilde v_iv_i,\tilde v_iw_i$, and a triangulation $\overline{F'}$ of $P'$ by adding the
edges $\tilde v'_iu_i,\tilde v'_ip_i,\tilde v_iw_i$. Now, the images of the edges $\tilde v_iv_i$ of $\overline{F}$, edges $\tilde v'_iv_i$, are not in $\overline{F'}$
(there the edges $\tilde v'_ip_i$ appear instead). This situation occurs $(s-k)$ times, so the number of compatible edges between $\overline{F}$ and $\overline{F'}$ is $3n-6-(s-k)$.

Conversely, suppose $P$ and $P'$ contain triangulations $\overline{F}$ and $\overline{F'}$ with $k'=3n-6-(s-k)$ compatible edges. If $\tilde{v_i}t_i$ is not in $\overline{F}$, then the edges $\tilde v_iv_i$
and $u_iw_i$ must be both in $\overline{F}$, because the edge $u_iw_i$ can be crossed only by the edges $t_iv_i$ and $t_i\tilde v_i$. However, in set $P'$, always the edge $\tilde v'_iv_i$
is crossed by the edge $u_iw_i$. Then, for each index $i$ such that $\tilde{v_i}t_i$ is not in $\overline{F}$, one of the edges $\tilde v_iv_i$ or $u_iw_i$ of $\overline{F}$ is not in $\overline{F'}$. Therefore, this situation can happen at most $(s-k)$ times, that is, the triangulation $\overline{F}$ must contain at least $k$ edges $\tilde{v_i}t_i$. But if $k$ segments $\tilde{v_i}t_i$ are disjoint, also their corresponding $v_it_i$ edges are disjoint. Therefore, $S$ has to contain at least $k$ disjoint segments.
\end{proof}
Finally, let us analyze the complexity of augmenting a plane subgraph $F$ of $D_n$ until obtaining a maximal plane subgraph. 
Since $F$ has $O(n)$ edges, the set of edges of $S(v)$ not crossing $F$ can be trivially found in $O(n^2)$ time. This directly implies an $O(n^3)$ algorithm to obtain a maximal plane graph containing $F$: For $i=1,\ldots ,n$, update $F$ by adding the edges of $S(v_i)$ non-crossing $F$, not in $F$. The following result implies that, if $F$ is connected, finding a maximal plane subgraph containing $F$ can be done in $O(n^2)$ time.

\begin{theorem}\label{thm:algorithm2}
Given a simple topological drawing of~$K_n$, a connected plane subgraph $F$, and a vertex $v$, we can find the edges from $v$ to $F$ not crossing $F$ in $O(n)$ time.
\end{theorem}
\begin{proof}
Notice that as $F$ is a plane graph, we can compute in linear time, for each vertex $w$ the clockwise order of the edges of $F$ incident to $w$, the faces of $F$, and for each face $f$, the clockwise cyclic list of edges and vertices found along its boundary.

Suppose first that the vertex $v$ is not in $F$, and let $vw_1$ be the first edge in the rotation of $v$ with one endpoint in $F$. The algorithm runs in three stages. In the first stage, it starts by finding the edge of $F$, edge $e_1=u_0u_1$, that intersects $vw_1$ closest to $v$ along $vw_1$. When the first intersection point occurs precisely at the vertex $w_1$, we take $e_1$ as the first edge of $F$ that follows, counterclockwise, to $w_1v$ in $S(w_1)$. 
Using the rotation system and its inverse, this edge $e_1$ of $F$ can be found in linear time, since $|F| \in O(n)$.
It also gives us the face~$f$ of $F$ containing the vertex $v$ inside.

For simplicity, let us suppose that $f$ is a bounded face and that the boundary of $f$ is a simple cycle, formed by the edges $e_1=u_0u_1,e_2=u_1u_2,e_m=u_{m-1}u_0$. We will later discuss the general case.
Notice that if the edges $vw_i,vw_j,vw_k$ are in this clockwise order in $S(v)$, their corresponding first crossing points $x_i,x_j,x_k$ with $F$  are found in a clockwise walk of the boundary of~$f$ in that same clockwise order. 
See the right bottom drawing of \fig{fig:order_along_face}.

In the second stage, the algorithm simulates a clockwise walk $x_1u_1,u_1u_2,\ldots ,u_{k-1}u_k,\ldots $ of the boundary of $f$ starting at point $x_1$, the first crossing point of  $vw_1$ with $e_1=u_0u_1$, and simultaneously a clockwise walk $vw_1,vw_2,\ldots ,vw_i,\ldots $ on the edges of the star $S(v)$, beginning with the edge $vw_1$. In each step, the algorithm makes progress in at least one of the two walks, by adding the following edge on the boundary to the boundary walk or passing to explore the following edge of $S(v)$.
In this process the algorithm will keep a list $\sigma $ with some of the explored edges of $S(v)$.

\begin{figure}[!htb]
\centering
\includegraphics[scale=0.8]{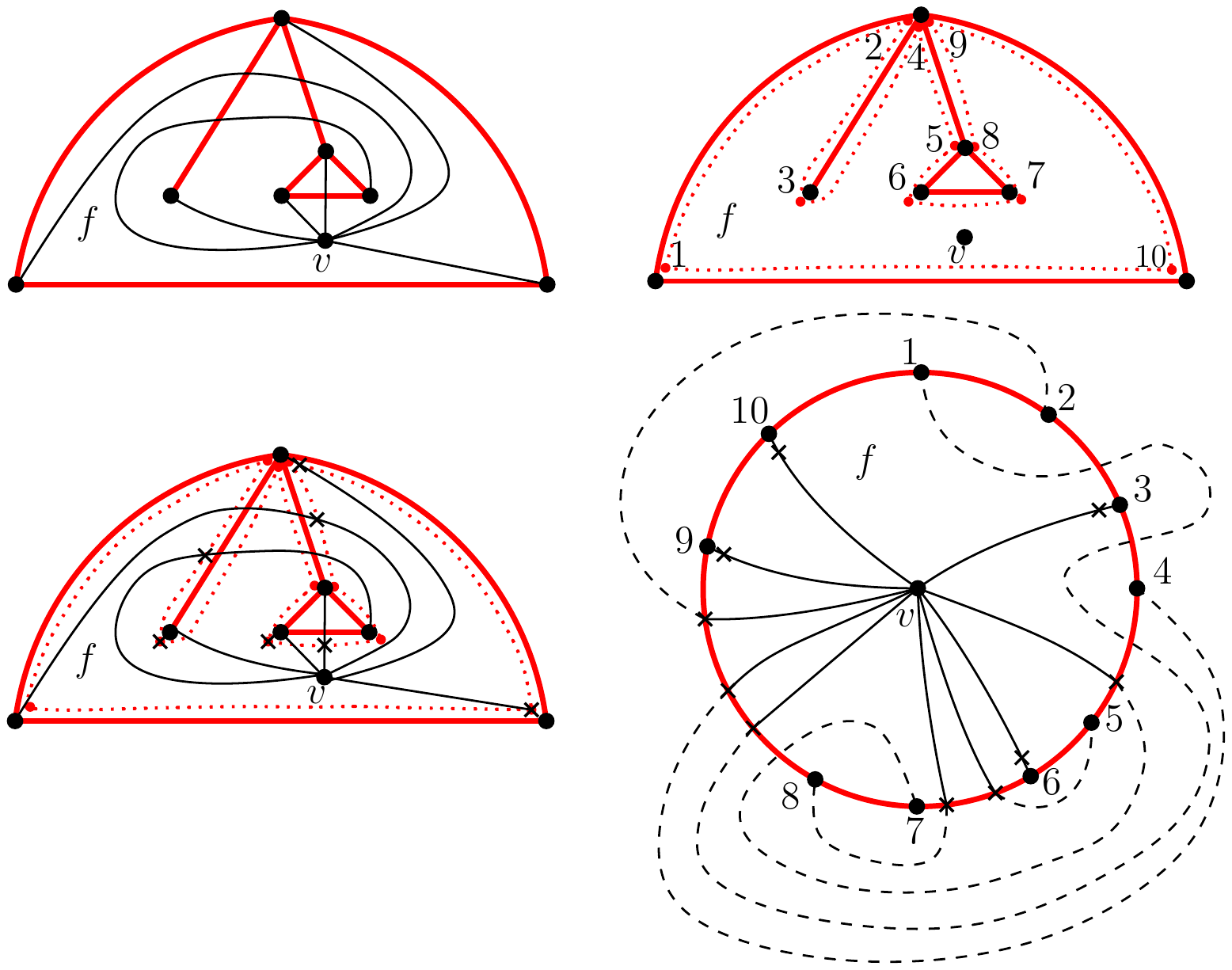}
\caption{Top left: A vertex $v$ inside the face $f$. Only the edges $vu_i$, with $u_i$ incident to the face $f$, can be uncrossed by $F$.
Top right: A clockwise walk along the boundary of the face $f$.
Bottom left: In a walk along the boundary of $f$, the first crossing points of the edges of $S(v)$ are found in the same order as the edges of $S(v)$. 
Bottom right: An equivalent drawing to the top left figure with the boundary of $f$ being a simple cycle. Some vertices, like $(2,4,9)$, can correspond to the same vertex of the first drawing.}
\label{fig:order_along_face}
\end{figure}

In a generic step, the edges $S_i=(vw_1,vw_2,\ldots ,vw_i)$ of $S(v)$, and the portion of the boundary of $f$, $W_k=(x_1u_1,u_1u_2,\ldots ,u_{k-1}u_k)$, have been visited, and the two following invariants hold:

(A) The first crossing of the edge $vw_i$ is not on $W_{k-1}=(x_1u_1,u_1u_2,\ldots ,u_{k-2}u_{k-1})$ (the walk $W_k$ minus its last edge).

(B) The list $\sigma $ contains an ordered list $(vu_{i_1},vu_{i_2},\ldots ,vu_{i_s})$ of the explored edges of $S(v)$ finishing  at some of the vertices $u_j,1\le j<k$, satisfying:

 (B1) All the explored edges of $S(v)$ not placed in $\sigma $ cross the boundary of $f$.

 (B2) The first crossing point of each edge $vu_{j}$ of $\sigma $ with the boundary of $f$ is either $u_{j}$ or is placed clockwise after $u_{j}$.

Initially, if $x_1$ is an interior point of the edge $e_1=u_0u_1$, then $W_k=(x_1u_1)$, $S_i=(vw_1,vw_2)$ and the list $\sigma $ is empty. If $x_1$ coincides with the vertex $u_0$, then $W_k=(u_0u_1)$, $S_i=(vw_1,vw_2)$ and the list contains the edge $vu_0$. In both cases invariants (A) and (B) are satisfied (the walk $W_{k-1}$ is empty or consists of only one vertex).

In this second stage the algorithm proceeds as follows:
\begin{itemize}
\item If $vw_i$ crosses the last edge of $W_k$, edge $e_k$, or if $w_i$ is not a vertex of $f$, iterate considering the clockwise successor $vw_{i+1}$ of $vw_i$ in the rotation of~$v$.

As the first crossing of $vw_i$ must be on the edge $e_k$ or a posterior edge $e_t,t>k$, also the first crossing of $vw_{i+1}$ must be on $e_k$ or a posterior edge. Thus invariant (A) is kept. On the other hand, observe that $\sigma $ does not change, $vw_i$ must not be included in $\sigma $ (it crosses $f$), and $W_k$ is not modified. Therefore invariant (B) is also kept.

\item If $vw_i$ does not cross $e_k$ and $w_i$ is a vertex of $f$, $w_i\neq u_k$, then, add the following edge $e_{k+1}$ on $f$ to $W_k$, keeping the same edge $vw_i$ of $S(v)$.

Invariant (A) is kept, because the first crossing point of $vw_i$ cannot be on $W_k$. Invariant (B) is also kept, because $\sigma $ is not modified.

\item If $vw_i$ does not cross $e_k$ and $w_i=u_k$,
    then, add $vw_i$ to the list $\sigma $, pass to explore the following edge $vw_{i+1}$ of $S(v)$ and add the following edge $e_{k+1}$ on $f$ to $W_k$.

Again, invariant (A) is kept, because the first crossing of $vw_{i+1}$ must be after $u_k$. On the other hand, the first crossing point of $vw_i$ is either $u_k$ or it is placed after $u_k$, hence property (B) is kept.

\end{itemize}

This second stage of the algorithm ends when all the edges of $S(v)$ and $f$  have been explored. The last edge of the boundary of $f$ being either $u_0x_1$  or $u_{m-1}u_0$.  Therefore, at the end,
invariant (B) implies that $\sigma$ will contain the uncrossed edges of~$S(v)$ plus some crossed edges $vu_i$ of $S(v)$ satisfying that the first crossing (on the boundary of $f$) is placed after the endpoint $u_i$ of that edge.

In each step of this stage, a new edge in the boundary of $f$, a new edge of $S(v)$, or both edges become explored.
As the number of edges in $f$ and in $S(v)$ is linear, this second stage of the algorithm runs in $O(n)$ time.

In the third stage, the algorithm repeats counterclockwise the above stage considering only the edges in $\sigma $. That means, it explores counterclockwise the boundary of $f$ (in the order $x_1u_0,$ $ u_0u_{m-1},\ldots $,  and counterclockwise the edges of $S(v)$ placed in $\sigma $ (so, in the order $vu_{i_s},$ $ vu_{i_{s-1}},\ldots $ ). In this third stage, in linear time, a new list $\overline{\sigma }$ is obtained. By invariant (B1), all the uncrossed edges of $S(v)$ have to be in $\overline{\sigma }$. And by invariant (B2), if $vu_i$ is in $\overline{\sigma }$, its first crossing point cannot be clockwise nor counterclockwise before $u_i$, so it has to be $u_i$. Therefore, $\overline{\sigma }$ will contain the uncrossed edges of $S(v)$.

In general, the boundary of face $f$ is not a simple cycle, some edges of $f$ can be incident to $f$ for both sides, so they appear twice in a walk along the boundary of $f$. However, this general case  can be transformed to the previous case by standard techniques, 
as done in ~\cite{fulek_ruiz_vargas} in their proof of the general case of Lemma~\ref{lem:fulek_ruiz_vargas}. In \fig{fig:order_along_face}, the bottom right figure shows how to transform the drawing of the top left figure, to obtain an equivalent drawing where the boundary of $f$ is a simple cycle. When the face $f$ is the unbounded face the algorithm is totally analogous.

Finally, let us consider the case when the vertex $v$ is in $F$. Then, vertex $v$ can be incident to several faces $f_1,\ldots ,f_l, l\ge 1$. Again, for simplicity, suppose that the boundary of each one of these faces is a simple cycle. For each face $f_i$, if $vw_{i_1},vw_{i_2}$ are the two edges incident to vertex $v$ in $f_i$, we can compute by the above method the uncrossed edges of $S(v)$ placed inside $f_i$, using only the edges of $S(v)$ placed clockwise between $vw_{i_1}$ and $vw_{i_2}$.
\end{proof}

\section{Conclusion}
In this paper, we considered maximal and maximum plane subgraphs of simple topological drawings of~$K_n$.
It turns out that maximal plane subgraphs have interesting structural properties.
These insights could be useful in improving the bounds on the number of disjoint edges in any such drawing, continuing this long line of research.

Also, algorithmic questions arise.
For example, Proposition~\ref{prop:ReducedRange} ensures that there are always two edges connecting a vertex $v$ to a not necessarily connected plane graph~$F$ in~$D_n$ without crossings.
Moreover, the set of edges of $S(v)$ not crossing $F$ can be trivially found in $O(n^2)$ time.
This leads to the following question.

\begin{problem}
Given a not necessarily connected plane graph $F$ in $D_n$, plus a vertex~$v$ not in $F$, can the edges of $S(v)$ incident to but not crossing $F$ be found in $o(n^2)$ time?
\end{problem}


\bibliographystyle{abbrv}
\bibliography{bibliography}

\begin{thebibliography}{10}

\bibitem{all_good_drawings}
B.~{\'{A}}brego, O.~Aichholzer, S.~Fern{\'{a}}ndez{-}Merchant, J.~Pummer,
  A.~P.~P. Ramos, G.~Salazar, and B.~Vogtenhuber.
\newblock All good drawings of small complete graphs.
\newblock In {\em EuroCG 2015}, pages 57--60, 2015.

\bibitem{bishellable15}
B.~M. {\'A}brego, O.~Aichholzer, S.~Fern{\'a}ndez-Merchant, D.~McQuillan,
  B.~Mohar, P.~Mutzel, P.~Ramos, R.~B. Richter, and B.~Vogtenhuber.
\newblock Bishellable drawings of {$K_n$}.
\newblock {\em SIAM J. Discrete Math.}, 32:2482--2492, 2015.

\bibitem{shellable_drawings}
B.~M. {\'{A}}brego, O.~Aichholzer, S.~Fern{\'{a}}ndez{-}Merchant, P.~Ramos, and
  G.~Salazar.
\newblock Shellable drawings and the cylindrical crossing number of {$K_n$}.
\newblock {\em Discrete Comput. Geom.}, 52(4):743--753, 2015.

\bibitem{fox_sudakov}
J.~Fox and B.~Sudakov.
\newblock Density theorems for bipartite graphs and related {R}amsey-type
  results.
\newblock {\em Combinatorica}, 29(2):153--196, 2009.

\bibitem{fulek_ruiz_vargas}
R.~Fulek and A.~J. Ruiz{-}Vargas.
\newblock Topological graphs: empty triangles and disjoint matchings.
\newblock In G.~D. da~Fonseca, T.~Lewiner, L.~M. Pe{\~{n}}aranda, T.~M. Chan,
  and R.~Klein, editors, {\em Symp. on Computational Geometry (SoCG 2013)},
  pages 259--266. {ACM}, 2013.

\bibitem{gioan}
E.~Gioan.
\newblock Complete graph drawings up to triangle mutations.
\newblock In D.~Kratsch, editor, {\em WG}, volume 3787 of {\em LNCS}, pages
  139--150. Springer, 2005.

\bibitem{Harborth74}
H.~Harborth and I.~Mengersen.
\newblock Edges without crossings in drawings of complete graphs.
\newblock {\em Journal of Combinatorial Theory, Series B}, 17:299--311, 1974.

\bibitem{klincsek}
G.~Klincsek.
\newblock Minimal triangulations of polygonal domains.
\newblock In P.~L. Hammer, editor, {\em Combinatorics 79}, volume~9 of {\em
  Annals of Discrete Mathematics}, pages 121--123. Elsevier, 1980.

\bibitem{segment_independent_set}
J.~Kratochv{\'i}l and J.~Ne{\v{s}}et{\v{r}}il.
\newblock {INDEPENDENT SET} and {CLIQUE} problems in intersection-defined
  classes of graphs.
\newblock {\em Comment. Math. Univ. Carolinae}, 31:85--93, 1990.

\bibitem{realizability_in_p}
J.~Kyn{\v{c}}l.
\newblock Simple realizability of complete abstract topological graphs in {P}.
\newblock {\em Discrete Comput. Geom.}, 45(3):383--399, 2011.

\bibitem{Kyncl2009}
J.~Kyn\v{c}l.
\newblock Enumeration of simple complete topological graphs.
\newblock {\em Eur. J. Comb.}, 30:1676--1685, 2009.

\bibitem{Kyncl2013}
J.~Kyn\v{c}l.
\newblock Improved enumeration of simple topological graphs.
\newblock {\em Discrete \& Computational Geometry}, 50:727--770, 2013.

\bibitem{Balko_Monotone}
J.~K. Martin~Balko, Radoslav~Fulek.
\newblock Crossing numbers and combinatorial characterization of monotone
  drawings of $k_n$.
\newblock {\em Discrete \& Comput. Geom.}, 53:107--143, 2015.

\bibitem{pach_solymosi_toth}
J.~Pach, J.~Solymosi, and G.~T{\'{o}}th.
\newblock Unavoidable configurations in complete topological graphs.
\newblock {\em Discrete {\&} Computational Geometry}, 30(2):311--320, 2003.

\bibitem{pach_toth}
J.~Pach and G.~T{\'{o}}th.
\newblock Disjoint edges in topological graphs.
\newblock In J.~Akiyama, E.~T. Baskoro, and M.~Kano, editors, {\em
  Combinatorial Geometry and Graph Theory, Indonesia-Japan Joint Conference
  (IJCCGGT 2003), Revised Selected Papers}, volume 3330 of {\em Lecture Notes
  in Computer Science}, pages 133--140. Springer, 2003.

\bibitem{pach_toth_2000}
J.~Pach and G.~Tóth.
\newblock Which crossing number is it anyway?
\newblock {\em Journal of Combinatorial Theory, Series B}, 80(2):225–246,
  2000.

\bibitem{pammer}
J.~Pammer.
\newblock Rotation systems and good drawings.
\newblock Master's thesis, Graz University of Technology, 2014.

\bibitem{rafla}
N.~H. Rafla.
\newblock {\em The good drawings $D_n$ of the complete graph $K_n$}.
\newblock PhD thesis, McGill University, Montreal, 1988.

\bibitem{Ringel}
G.~Ringel.
\newblock Extremal problems in the theory of graphs.
\newblock In M.~Fiedler, editor, {\em Proceedings of the Symposium in Theory of
  Graphs and Its Applications,Smolenice}, pages 85--90, 1963.

\bibitem{many_disjoint}
A.~J. Ruiz{-}Vargas.
\newblock Many disjoint edges in topological graphs.
\newblock {\em Comput. Geom.}, 62:1--13, 2017.

\bibitem{saalfeld}
A.~Saalfeld.
\newblock Joint triangulations and triangulation maps.
\newblock In D.~Soule, editor, {\em Proceedings of the Third Annual Symposium
  on Computational Geometry, Waterloo, Ontario, Canada, June 8-10, 1987}, pages
  195--204. {ACM}, 1987.

\bibitem{Schaefer2013}
M.~Schaefer.
\newblock The graph crossing number and its variants: A survey.
\newblock {\em The electronic journal of combinatorics}, DS21, 2017.

\bibitem{suk}
A.~Suk.
\newblock Disjoint edges in complete topological graphs.
\newblock {\em Discrete {\&} Computational Geometry}, 49(2):280--286, 2013.

\end{thebibliography}

\end{document}